\documentclass[11pt]{article}
\usepackage{amsmath,amssymb,amsfonts}
\begin{document}
\newtheorem{theorem}{Theorem}[section]
\newtheorem{lemma}[theorem]{Lemma}
\newtheorem{definition}[theorem]{Definition}
\def\emptyset{\varnothing}
\def\setminus{\smallsetminus}
\def\loc{{\mathrm{loc}}}
\def\opp{{\mathrm{opp}}}
\def\id{{\mathrm{id}}}
\def\A{{\mathcal{A}}}
\def\B{{\mathcal{B}}}
\def\C{{\mathcal{C}}}
\def\D{{\mathcal{D}}}
\def\N{{\mathbb{N}}}
\def\Q{{\mathbb{Q}}}
\def\R{{\mathbb{R}}}
\def\Z{{\mathbb{Z}}}
\def\a{{\alpha}}
\def\e{{\varepsilon}}
\def\la{{\lambda}}
\def\isom{{\cong}}
\newcommand{\Hom}{\mathop{\mathrm{Hom}}\nolimits}
\def\qed{{\unskip\nobreak\hfil\penalty50
\hskip2em\hbox{}\nobreak\hfil$\square$
\parfillskip=0pt \finalhyphendemerits=0\par}\medskip}
\def\proof{\trivlist \item[\hskip \labelsep{\bf Proof.\ }]}
\def\endproof{\null\hfill\qed\endtrivlist\noindent}

%%%%%%%%%%%%%%%%%%%%%%%%%%%%%%
\title{A remark on gapped domain walls between topological phases}
\author{
{\sc Yasuyuki Kawahigashi}\footnote{Supported in part by 
Research Grants and the Grants-in-Aid
for Scientific Research, JSPS.}\\
{\small Graduate School of Mathematical Sciences}\\
{\small The University of Tokyo, Komaba, Tokyo, 153-8914, Japan}
\\[0,05cm]
{\small and}
\\[0,05cm]
{\small Kavli IPMU (WPI), the University of Tokyo}\\
{\small 5-1-5 Kashiwanoha, Kashiwa, 277-8583, Japan}\\
{\small e-mail: {\tt yasuyuki@ms.u-tokyo.ac.jp}}}
\maketitle{}
\begin{abstract}
We give a mathematical definition of a gapped domain wall between
topological phases and a gapped boundary of a topological phase.
We then provide answers to some recent questions
studied by Lan, Wang and Wen in condensed matter physics
based on works of Davydov, M\"uger, Nikshych and Ostrik. 
In particular, we identify their tunneling  matrix 
and a coupling matrix of Rehren, and show that
their conjecture does not hold.
\end{abstract}

\section{Introduction}

The modern subfactor theory was initiated by Jones \cite{J}
and deep connections to chiral conformal field theory
have been known.  (See \cite{EK} for a general subfactor
theory and \cite{K} for connections to conformal field theory.)
A natural language for studying such connections is
that of a tensor category.  A unitary fusion category
\cite{ENO} is particularly natural for subfactor theory
and a unitary modular tensor category \cite{BK}, \cite{T}
is particularly natural for the operator algebraic
study of chiral conformal field theory \cite{KLM}.

Recently, studies of topological phases, condensation,
gapped domain walls, and gapped boundaries have been
made in the language of modular tensor categories
\cite{BS}, \cite{FSV}, \cite{HW1}, \cite{HW2}, \cite{KK},
\cite{Ko}, \cite{LW}, \cite{LWW}.  The mathematical 
structures appearing in these studies are the same
as those appearing in subfactors and
chiral conformal field theory.  One dictionary between
$\a$-induction for subfactors and condensation has
been supplied by us in \cite[page 440]{Ko}.

In this note, we further give a correspondence between
studies of subfactors and
chiral conformal field theory and those of
gapped domain walls and gapped boundaries.  In particular,
we answer some question raised in \cite{LWW} and
show that a conjecture in \cite{LWW} does not hold.

A part of this work was done during a stay at Universit\`a di
Rome ``Tor Vergata''.  The author thanks for the hospitality.
The author also thanks R. Longo, Y. Ogata and Z. Wang for
stimulating discussions related to the topics.
The author also thanks M. Barkeshli, M. Bischoff, J. Fuchs,
L. Kong, T. Lan, C. Schweigert, Y. Wan, J. Wang,
and X.-G. Wen for comments on the first version of this paper.

\section{Preliminaries}

We start with a unitary modular tensor category $\C$ as in 
\cite{BK}, \cite{ENO}, \cite{T}.  
Each object $\la$ in $\C$ has $\dim \la\in\R_+$ and 
$\C$ has $\dim\C=\sum_\la (\dim \la)^2$, where $\la$ is
a simple object in $\C$.  As in
\cite{BS}, \cite{HW1}, \cite{HW2}, \cite{Ko}, 
\cite{LW}, \cite{LWW},
a topological phase is represented by a unitary modular
tensor category.
Note that we have the additive central charge of $\C$
in $\Q/8\Z$ as after Lemma 5.25 in \cite{DMNO}.

The following definition was introduced in 
\cite[Definition 5.1]{DMNO}, where $\D^\opp$ means the
modular tensor category with the braiding reversed.

\begin{definition}{\rm
We say that unitary modular tensor categories $\C$
and $\D$ are {\it Witt equivalent} if
$\C\boxtimes \D^\opp$ is equivalent to the quantum double
of some unitary fusion category.
}\end{definition}

The following is due to \cite[Theorem 6.1]{L}, 
\cite[Theorem 4.9]{LR}.

\begin{definition}{\rm
A $Q$-{\it system} $(\theta, v, w)$ in a unitary modular tensor
category $\C$ is a triple of an object $\theta$ in $\C$,
$v\in\Hom(\id,\theta)$,
$w\in\Hom(\theta,\theta^2)$ satisfying the following identities.
\begin{align*}
(v^*\times 1) w &=(1\times v^*)w \in\R_+,\\
(w\times 1)w&=(1\times w)w.
\end{align*}
We say a $Q$-system $(\theta, v, w)$ is {\it irreducible}
when $\Hom(\id,\theta)$ is 1-dimensional.
We say a $Q$-system $(\theta, v, w)$ is {\it local}
if we have $w=\e(\theta,\theta)w$, where $\e$ denotes the
braiding of $\C$.
}\end{definition}

An irreducible local $Q$-system is the same as a
{\it condensable algebra} in \cite[Definition 2.6]{Ko}.
If we have an irreducible local $Q$-system $(\theta, v, w)$
in a unitary modular tensor category $\C$, we can realize $\C$
as a category of endomorphisms of a type III factor $N$ and
apply the machinery of the $\alpha$-induction as in
\cite{LR}, \cite{X}, \cite{BE1}, \cite{BE2}, \cite{BE3},
\cite{BEK1}, \cite{BEK2}, \cite{BEK3}.
The study of condensation in the setting of a modular
tensor category is parallel to the study of $\a$-induction
as in the table in \cite[page 440]{Ko} supplied by us.
(The $\a$-induction also works in a more abstract setting
of tensor categories.  See \cite[Section 5]{O}.)

If we have an irreducible local $Q$-system $(\theta, v, w)$
in a modular tensor category $\C$, we have an
extension modular tensor category $\tilde\C$ as the
{\it ambichiral system} in the sense of \cite[page 741]{BEK2}
based on $\alpha$-induction.
(It was proved to be a modular tensor category in
\cite[Theorem 4.2]{BEK2}.)
If a modular tensor category $\C$ is a representation
category of a completely rational local conformal net
$\{A(I)\}$ in the sense of \cite{KLM}, then an 
irreducible local $Q$-system produces an extension
$\{\tilde\A(I)\}$ of $\{\A(I)\}$, and $\tilde \C$ is
the representation category of $\{\tilde\A(I)\}$.
In \cite[Definition 1.1]{KO}, a  local $Q$-system 
is called a commutative associative algebra $A$ in 
a modular tensor category $\C$, and the extension
modular tensor category $\tilde\C$ here is denoted
by $\C_A^\loc$ in \cite{KO}.

The following is introduced in \cite[Definition 4.6]{DMNO}.

\begin{definition}{\rm
A $Q$-system $(\theta,v,w)$ in a unitary modular tensor
category $\C$ is called {\it Lagrangian}
if we have $(\dim\theta)^2=\dim\C$.
}\end{definition}

Note that this condition is equivalent to the one that
the extension category $\tilde \C$ arising from $(\theta,v,w)$
is trivial in the sense $\dim\tilde\C=1$.

The following has been proved in \cite[Section 4]{DMNO}.
Here our ``quantum double'' is often called the
Drinfeld center in literature, and has been studied
in \cite{I} well in the setting of the Longo-Rehren 
subfactor \cite{LR}.

\begin{theorem}\label{double-ch}
A unitary modular tensor category has an irreducible local
Lagrangian $Q$-system if and only if it is the quantum double
of some unitary fusion category.
\end{theorem}

Since we are interested in only unitary categories, the ``if'' part
of the above theorem follows from the description of the Longo-Rehren
subfactor in \cite{I}, and the ``only if'' part follows from
\cite[Corollary 3.10]{BE4}.

\section{Gapped domain walls and $Q$-systems}

Based on arguments in \cite{LWW}, we make the following 
definition.

\begin{definition}{\rm
For two unitary modular tensor categories $\C$ and $\D$,
an irreducible Lagrangian local $Q$-system in
$\C\boxtimes \D^\opp$ is called a {\it gapped domain wall}
between $\C$ and $\D$.  If $\D$ is trivial in the sense
$\dim \D=1$, we say the $Q$-system
is a {\it gapped boundary} of $\C$.
}\end{definition}

Such a $Q$-system is of the form 
$(\bigoplus Z_{\la,\mu} \lambda\boxtimes\mu, v, w)$,
where $\la$ and $\mu$ label the  simple objects of $\C$
and $\D$, respectively.  The unitary modular tensor categories
here represent {\it topological phases}.
The gapped domain wall here is the same as the Witt
trivialization of $\C\boxtimes \D^\opp$ 
in the sense of \cite[Definition 2.16]{FSV}. See \cite{FSV}
for a formulation of Witt trivialization in connection
to defect Wilson lines.
In \cite{LWW}, the matrix $Z$ is called a {\it tunneling matrix}.
Note that a gapped domain wall between $\C$ and $\D$ is
the same as a gapped boundary of $\C\boxtimes \D^\opp$.
This is compatible with the ``folding trick'' in \cite{KK}
mentioned in \cite{LWW}.  The above definition is also
compatible with the works \cite{HW1}, \cite{HW2}, \cite{LW}.
(See explanations on condensation in \cite{Ko}.
Also see earlier works \cite{BJQ}, \cite{Le} for the
abelian phases.)

If $\C$ and $\D$ are the representation categories of local
conformal nets $\{\A(I)\}$, $\{\B(I)\}$, respectively, a
domain gapped wall between them 
is the situation considered in \cite{R}.  In the context
of \cite{R}, the matrix $Z$ was called a {\it coupling matrix}.

We prepare a simple lemma.

\begin{lemma}\label{easy}
Let $(\theta, v, w)$ be a local $Q$-system in a modular
tensor category $\C$.  For simple objects $\la,\mu\in\C$,
we have $\langle \theta,\la\mu\rangle\ge
\langle \theta,\la\rangle
\langle \theta,\mu\rangle$, where
$\langle\la,\mu\rangle=\dim\Hom(\la,\mu)$.
\end{lemma}

\begin{proof}
The object $\theta$ can be regarded as a dual canonical
endomorphism of a subfactor $N\subset M$.
Let $\theta=\bar\iota\iota$ for $\iota:N\hookrightarrow M$.
Then we have the following by the Frobenius reciprocity.
\begin{align*} 
\langle \bar\iota\iota,\la\mu\rangle&=
\langle \iota\bar\mu,\iota\la\rangle,\\
\langle \bar\iota\iota,\la\rangle
&=\langle\iota,\iota\la \rangle,\\
\langle \bar\iota\iota,\mu\rangle&
=\langle \iota,\iota\bar\mu\rangle.
\end{align*} 
These imply the desired inequality.
\end{proof}

We now have the following result corresponding to the
main result in \cite{LWW} where the three conditions 
below are given.

\begin{theorem}\label{LWW-cond}
Let $(\bigoplus Z_{\la,\mu} \lambda\boxtimes\mu, v, w)$
be a gapped domain wall between $\C$ and $\D$.
Let $S^\C, T^\C, N^\C$ [resp. $S^\D, T^\D, N^\D$] be
the $S$-matrix, the $T$-matrix, the fusion rule coefficients
of $\C$ [resp. $\D$].  We then have the following.
\begin{itemize}
\item[(1)] $Z_{\la\mu}\in\N$.
\item[(2)] $S^\C Z=ZS^\D$, $T^\C Z=ZT^\D$.
\item[(3)] $Z_{\la\mu}Z_{\la'\mu'}
\le\sum_{\la'',\mu''}(N^\C)_{\la\la'}^{\la''}Z_{\la''\mu''}
(N^\D)_{\mu\mu'}^{\mu''}$.
\end{itemize}
\end{theorem}

\begin{proof}
Condition (1) is obvious.

By Theorem \ref{double-ch}, we know that 
$\C$ and $\D$ are Witt equivalent.  Then
\cite[Proposition 3.7, Corollary 3.8]{DNO} specifies
the form of the $Q$-system.  Then \cite[Theorem 6.5]{BE4}
implies the desired commutativity (2).

Since $Z_{\la''\mu''}=
\langle \theta(\bar\la''\boxtimes\id),
\id\boxtimes\mu''\rangle$, we have
$$\sum_{\mu''}Z_{\la''\mu''}(N^\D)_{\mu\mu'}^{\mu''}=
\langle \theta(\bar\la''\boxtimes\id),\id\boxtimes \mu\mu'\rangle
=\langle \theta(\id\boxtimes\bar\mu'\bar\mu),
\la''\boxtimes\id\rangle.$$
Then the right hand side of (3) is equal to the following.
$$\langle \theta(\id\boxtimes\bar\mu'\bar\mu),
\la\la'\boxtimes\id\rangle=\langle\theta,
\la\la'\boxtimes\mu\mu'\rangle.$$
Since the left hand side of (3) is equal to
$\langle \theta,\la\boxtimes\mu\rangle
\langle \theta,\la'\boxtimes\mu'\rangle$,
Lemma \ref{easy} gives the desired inequality.
\end{proof}

In the case where $\C$ and $\D$ are the representation
categories of completely rational local conformal nets,
condition (2) was obtained by M\"uger.  
(See \cite[Theorem 3.1]{KL2}.)

The first half of Question (Q1) in \cite{LWW} asks when
we have a gapped domain wall between $\C$ and $\D$.  As we
have seen, the
answer is given by Theorem \ref{double-ch}.  That is,
we have one if and only if $\C$ and $\D$ are Witt
equivalent.  This necessity was first pointed out in
\cite[pagr 560]{FSV}.

The second half of Question (Q1) in \cite{LWW}
asks how many gapped domain walls we have between
$\C$ and $\D$.  This is a finite combinatorial problem
on $Q$-systems as shown in \cite{IK} since we have only
finitely many matrices $(Z_{\la\mu})$ for a given $\C$
and $\D$, but counting in general can be highly complicated, and
we have no easy general methods.

It is conjectured in \cite{LWW} that these three conditions
are sufficient to have a gapped domain wall, but the example
in \cite{D2} shows that this is not true.  The examples
in \cite{D2} shows that the charge conjugation modular invariant
of the quantum double of some finite group does not have
a corresponding $Q$-system structure.  Note that the
charge conjugation modular invariant for $\C=\D$
obviously satisfies
(1) and (2) above, and it is also easy to see that it 
satisfies (3).  So this is an example of a matrix $Z$ satisfying
(1), (2) and (3) above, but it does not correspond to any
gapped domain wall.  (Note that the identity matrix in \cite{D2}
corresponds to the charge conjugation matrix here due to
a different convention.)  
Our formulation in terms of a $Q$-system
is the same as the one in terms of a condensable
algebra in the sense of \cite{Ko}, so the hope in \cite{LWW}
to have a condensable algebra from some data satisfying 
(1), (2) and (3)  does not work.
From various classification results
related to modular invariants in \cite{EP}, \cite{KL1},
\cite{KL2}, it looks impossible to have a sufficient condition
for a gapped domain wall between $\C$ and $\D$ to exist
simply by working on the tunneling matrix $(Z_{\la\mu})$,
and it is expected we have many other examples of 
matrices $Z$ satisfying (1), (2) and (3) which do not
correspond to gapped domain walls.
Also as remarked in \cite{LWW} based on \cite{D1},
the matrix $(Z_{\la\mu})$ does not have enough information to
recover the $Q$-system.  This is also compatible with
a well-known fact in the study of $Q$-systems. 
Also see \cite[Example 5.8]{BKL}.

\end{document}